\documentclass[journal]{IEEEtran}

\ifCLASSINFOpdf

\else

\fi

\usepackage{tikz}
\usepackage{graphicx}
\usepackage{citesort}
\usepackage{amsmath, bm}
\usepackage{amsthm}
\usepackage{amssymb, epsfig}
\usepackage{epstopdf}
\usepackage{subfigure,algorithmic,float}
\usepackage{lipsum}
\usepackage{textcomp}
\usepackage{color}
\usepackage{pgfplots}
\usepackage{filecontents}
\usepackage{xcolor}
\usepackage{mathtools}
\usepackage{color}
\usepackage{stfloats}
\pgfplotsset{compat=newest}
\usetikzlibrary{decorations.markings}
\usetikzlibrary{patterns}
\usetikzlibrary{decorations.pathreplacing,angles,quotes}
\theoremstyle{plain}

\newtheorem{lemma}{Lemma}
\newtheorem{prop}{Proposition}

\newtheorem{remark}{Remark}
\theoremstyle{definition}

\providecommand{\keywords}[1]{\textbf{\textit{Index Terms ---}} #1}

\begin{document}

\title{\Huge{Covert Wireless Communications under Quasi-Static Fading with Channel Uncertainty}}

\author{\large{Khurram~Shahzad, \IEEEmembership{Member, IEEE,} and Xiangyun~Zhou, \IEEEmembership{Senior Member, IEEE} }

%\thanks{This work was supported by the Australian Research Council's Discovery Projects under Grant DP180104062. }
\thanks{The authors are with the Research School of Electrical, Energy and Materials Engineering, The Australian National University, Canberra, ACT, Australia. (Emails: \{khurram.shahzad, xiangyun.zhou\}@anu.edu.au). }
}

\maketitle

\begin{abstract}
Covert communications enable a transmitter to send information reliably in the presence of an adversary, who looks to detect whether the transmission took place or not. We consider covert communications over quasi-static block fading channels, where users suffer from channel uncertainty. We investigate the adversary Willie's optimal detection performance in two extreme cases, i.e., the case of perfect channel state information (CSI) and the case of channel distribution information (CDI) only. It is shown that in the large detection error regime, Willie's detection performances of these two cases are essentially indistinguishable, which implies that the quality of CSI does not help Willie in improving his detection performance. This result enables us to study the covert transmission design without the need to factor in the exact amount of channel uncertainty at Willie. We then obtain the optimal and suboptimal closed-form solution to the covert transmission design. Our result reveals fundamental difference in the design between the case of quasi-static fading channel and the previously studied case of non-fading AWGN channel.
\end{abstract}

\keywords{\small Physical layer security, covert wireless communications, channel uncertainty, channel training, quasi-static fading channel. \normalsize}

\IEEEpeerreviewmaketitle

\section{Introduction}
\subsection{Background}
The positives of digital explosion have resulted in the negatives of security concerns, both in business and private domains. Governments and corporations are determined to ensure that their digital assets are properly protected, so that consumers can access the information and resources in confidence. The security and privacy of information transmitted over the air has always been a concern for wireless system designers, with a recently renewed interest owing to the advances and innovations in wireless technologies and their widespread use in everyday activities. Traditional techniques in security of wireless transmissions focus on maintaining the message confidentiality, looking to develop approaches such that the message content is only accessible to the intended receiver. In this regard, the standard practices in cryptography \cite{stinson2005cryptography} look to encode the message in such a way that the eavesdropper / unintended receiver is unable to decode the message, at least not without significant computation. On the other hand, physical layer security \cite{bloch_book,sean_book} exploits the imperfections and uncertainties of the physical channel, such as thermal noise and interference, to achieve security and privacy.

To augment the existing approaches to security, a new viewpoint has recently been proposed termed as Covert Communications or Low Probability of Detection Communications \cite{commag15bash,shihao2019mag}. Covert communications propose to hide the very existence of the wireless transmission itself. In contemporary social and political backdrops, situations exist where in addition to protecting the information content of the transmission, it is imperative to hide the transmission. For example, hiding communications in a sensitive or hostile environment is of paramount importance to military and law enforcement agencies. On the other hand, landing of sensitive information, e.g., pertaining to health issues or financial transactions of an individual, in the wrong hands can be exploited and is highly undesirable. In above mentioned and many other potential scenarios, covert communications offer a viable pathway which can be used in conjunction with existing security approaches to enhance user privacy. Recent research efforts in the domain of covert communications have explored different problems in this field, ranging from establishing the achievable fundamental limits to exploiting any uncertainties including channel noise and interference. The fundamental limits over additive white Gaussian noise (AWGN) channels have been derived in \cite{bash_jsac}, where the authors provide a square root law on the amount of information that can be transmitted covertly and reliably.

One of the major approaches to covert communications stems from information theory where researchers have focused on characterizing the covert capacity of different communication scenarios. Initial works extended the square root law of \cite{bash_jsac} to binary symmetric channels (BSCs) \cite{jaggi_ISIT_13}, discrete memoryless channels (DMCs) \cite{wang_TIT} and multiple access channels (MACs) \cite{bloch_ISIT_16}. Based on the principle of channel resolvability, \cite{bloch_TIT} developed a coding scheme which improves upon the size of required key shared between the transmitter and the receiver, while identifying the fundamental limits of covert communications in terms of optimal asymptotic scaling of the message and key size. \cite{tahmasbi2018first} studies the first and second order asymptotics of covert communication when measuring covertness in terms of relative entropy and in terms of variational distance between the channel output distributions, while the authors in \cite{tahmasbi2019covert} study the problem of covert and secret key generation over a state-dependent DMC in the presence of an active adversary. Apart from these works on discrete channels, covert communications under continuous-time Poisson channel and under spectral mask constraints have been considered in \cite{wang2018continuous} and \cite{zhang2019undetectable}, while \cite{wang2018spawc} considered the problem of covert communications over a continuous-time additive Gaussian noise channel, where it has been shown that under no bandwidth constraint, the covert communication capacity of the channel is positive.

Although under the square root law, the average number of bits transmitted covertly per channel use asymptotically reaches zero, it has been shown that a positive covert rate is achievable in the presence of uncertainties at the adversary. These include the situations of Willie's uncertainty in the knowledge of noise power \cite{lee_undetect,goeckel_noise_2016,biao_cc} and presence of a continuously transmitting jammer in the communication environment \cite{tamara_jammer_2017}. The case of additional friendly nodes generating artificial noise, causing confusion at Willie regarding the received signal statistics, is presented in \cite{soltani2018covert}, while achieving covertness with the aid of artificial noise transmitted by a full-duplex receiver was demonstrated in \cite{shahzad2018achieving}. Furthermore, covert communications under relay networks have been considered in \cite{hu2018covert_relay,hu2019covert} while a study on covert communications in the presence of a Poisson distributed field of interferers has been presented in \cite{he2018covert}. More recently, \cite{zheng2019multi} considered the performance of multi-antenna covert communications in random wireless networks, while the optimality of Gaussian signalling for covert communications under the asymmetry of Kulback-Leibler divergence was discussed in \cite{yan2019gaussian}.

The above-mentioned works consider covert communications under the assumption of an infinite number of channel uses. However, limited storage resources and requirements of quick updates in modern communication systems require a finite, sometimes small, number of channel uses, and hence the results in the infinite blocklength regime do not hold anymore. Covert communications under finite blocklength have also been previously considered in the literature. The authors in \cite{yan2019delay} consider achieving covertness under AWGN channels where the maximum number of allowed channel uses is constrained. Furthermore, \cite{shu2019delay} has considered achieving covertness under strict delay requirements using a full-duplex receiver, where it has been shown that in contrast to asymptotically infinite channel uses, a fixed power artificial noise transmission helps improve covert communications. The authors in \cite{shahzad2019covert} have analyzed covert communications under finite blocklength in the presence of a multi-antenna Willie, while covert communications over slow fading channels under finite blocklength has been considered in \cite{tang2018covert}, providing an upper bound on the total power satisfying a desired probability of detection by the adversary.

\subsection{Scenario, Approach and Contributions}
In this work, we consider achieving covert communications under finite blocklength where both Bob and Willie have imperfect knowledge of their respective channels from Alice. To help Bob estimate his channel, Alice transmits publicly known pilot symbols which also facilitate timing and carrier synchronization. We note here that the transmission of public pilots does not affect the scenario of covert transmissions since Alice is looking to hide its \textit{data} transmission to Bob despite Willie being aware of her presence. Intuitively, it is clear that the higher the training budget, the lower will be the channel estimation error, resulting in a higher throughput. This pilot transmission, on the other hand, also enables Willie to estimate his channel from Alice, improving his capability to detect any covert transmission. While the impact of imperfect channel knowledge on the throughput performance and schemes to alleviate the effects of these imperfections has been studied in detail in prior literature \cite{xia2005effect,yoo2006capacity,rezki2012ergodic}, the impact on Willie's detection performance in the domain of covert communications is much less explored. The use of pilot symbols causing less or more training at the legitimate receiver offers an interesting trade-off since helping the receiver obtain a better channel estimate also improves the detection ability of Willie.

Covert communications under imperfect channel knowledge has been previously considered in \cite{shahzad_vtc}, where under asymptotically infinite blocklength, the variance of channel uncertainty at the users has been incorporated in the analysis. The authors in \cite{wang2018covert} present a scheme where covertness is achieved with the help of a full-duplex relay, and users suffer from channel uncertainty. More recently, \cite{xu2019pilot} presented an analysis of channel estimation design in covert communications, where the number of training channel uses to maximize the effective signal-to-noise ratio at the covert link is optimized. While \cite{shahzad_vtc} and \cite{wang2018covert} present their analysis under infinite blocklength assumptions, additional sources of uncertainty in the form of an additional information receiver and an artificial noise transmitting relay, respectively, have been considered in these works to achieve covertness. Although similar to our considered scenario, \cite{xu2019pilot} presents the analysis under a finite blocklength, the authors consider an AWGN channel for Willie, whereas we consider quasi-static fading channels for both Bob and Willie. Furthermore, \cite{xu2019pilot} advocates the use of equal powers during the training and data transmission phases, while we first establish the best detection performance at Willie and then optimize the data transmission power to maximize the covert throughput under certain covertness requirements. While the above mentioned works specifically rely on and exploit the channel uncertainty at Willie to achieve covertness, we show that in scenarios pertinent to covert communications, where the transmit power levels are generally low, Willie's channel knowledge does not play as an important role as considered in the prior work, and hence, we are able to provide a unified approach to covert transmission design regardless of the exact amount of channel uncertainty at Willie.

\textit{The contribution of this work is two-fold.} First, we analytically derive Willie's optimal detection performance. Focusing on the large detection error regime\footnote{The large detection error regime at Willie refers to the scenarios most relevant to covert communications where Willie's detection error probability is desired to be as close to $1$ as possible. In the problem formulation, this refers to the design of communication schemes satisfying strict covertness requirements.}, which is most relevant for covert transmissions design, we show that Willie's detection performance is extremely insensitive to the accuracy of his channel knowledge. This constitutes a very useful result because it implies that the design for covert transmission does not rely on the knowledge of Willie's channel estimate, nor it takes into consideration as to how Willie obtains this channel estimate. Hence, as long as Willie is forced to stay in the large detection error regime by an appropriate transmission strategy, the accuracy of Willie's channel knowledge has almost no impact on his detection capability.

Second, in order to maximize the communication throughput under a given covertness constraint, we provide the optimal choice of the number of data symbols and data transmission power to be used by Alice. While the solution to the optimal problem at Alice requires a numerical search, we also provide a suboptimal closed-form solution, which offers a trade-off between obtaining a closed-form solution and a moderate reduction in the achievable performance. Our work reveals a fundamental difference in covert transmission design between the case of non-fading AWGN channel studied in \cite{yan2019delay} and the quasi-static fading channels. Specifically, for AWGN channels where noise variation is the determining factor, the more data symbols used by Alice, the higher the covert throughput is \cite{yan2019delay}. For quasi-static fading channels, however, the channel variation becomes much more important than noise variation, which holds true as long as there is a sufficient number of data symbols in the transmission \cite{durisi2016toward}. Under this condition, our outage-based performance analysis shows that the less the number of data symbols used by Alice, the higher the covert throughput is.

\subsection{Paper Organization}
The rest of this paper is organized as follows: Section II provides details of our communication scenario, considered channel estimation and training and the assumptions used in this paper. Section III explains the detection at Willie under perfect channel state information (CSI) and channel distribution information (CDI) only scenarios, and establishes the equivalence of these two cases for low transmit powers at Alice. In Section IV, we analyze the covertness achieved by Alice, addressing the optimal design of data transmit powers and channel uses to maximize the covert throughput under a given covertness constraint. Section V provides numerical results validating our analysis and giving further design insights. Finally, the paper is concluded in Section VI.

\section{System Model}

\begin{figure}
\centering
	\includegraphics[width=0.8\linewidth]{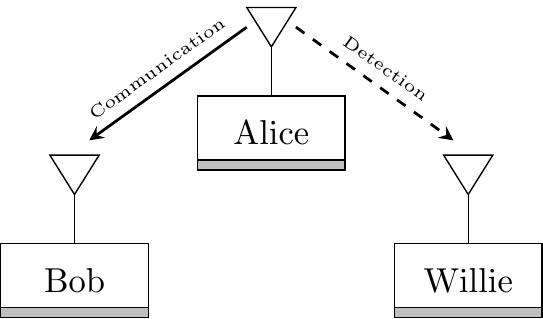}
	\caption{Covert communications model under consideration.}
	\label{system_model_tr}
\end{figure}

We consider a covert communication scenario, as shown in Fig. \ref{system_model_tr}, where the transmitter, Alice, desires to send information to the receiver, Bob, in presence of an adversary, Willie, whose job is to detect whether any transmission by Alice took place or not. Alice, Bob and Willie are assumed to have a single antenna each. The complex Gaussian noise at Bob and Willie's receivers is denoted by $n_b \sim \mathcal{CN}(0,\sigma_b^2)$ and $n_w \sim \mathcal{CN}(0,\sigma_w^2)$, respectively. We follow the common assumption that a secret is shared between Alice and Bob \cite{bash_jsac,shahzad2018achieving}, which is unknown to Willie but lets Bob know when Alice transmits a covert message. Employing random coding arguments, Alice generates codewords by independently drawing symbols from a zero-mean complex Gaussian distribution, where the codebook is known to Alice and Bob only. We define a communication slot as a block of time in which the transmission of a message from Alice to Bob is complete. When Alice transmits in a slot, she transmits the codeword corresponding to her covert message.

\subsection{Channel Model}
We consider the channels from Alice to Bob, and Alice to Willie to be quasi-static Rayleigh fading channels, where the effect of fading is modelled by a fading coefficient, $h_{k}$, and $k$ is either $b$ (Bob), or $w$ (Willie). Here, $h_{k}$ follows a circularly symmetric complex Gaussian (CSCG) distribution, with zero mean and unit variance, i.e., $h_{k} \sim \mathcal{CN}(0,1)$.  Due to the quasi-static fading assumption, the fading coefficients remain constant in one slot (i.e., one coherence interval), and change independently from one slot to the next. It is assumed here that the slot boundaries are synchronized among all parties. Due to the independent change of fading coefficients among slots, the focus is on one given slot, as the knowledge of previous slots does not help Willie in improving his detection performance \cite{shahzad_vtc,tamara_jammer_2017}. In general, the knowledge of CSI at Bob and Willie depends on the methods they use to acquire the CSI. The extreme cases include: the case of perfect CSI, where the instantaneous fading coefficients are exactly known; and the case of CDI only, where only the long-term statistics of the fading coefficients are known.

\subsection{Training-Based Transmission}
 Alice transmits publicly known pilot symbols periodically at the beginning of every slot, whereas, covert data transmission only occurs in a secretly chosen slot, which is only known to Alice and Bob. It is assumed that the transmission of a covert message is completed within a slot chosen secretly by Alice and Bob. Each slot consists of $N$ symbols, where the first $N_T$ symbols serve as pilots, and are transmitted using power $P_T$. Depending on whether or not covert data transmission happens in the current slot, data symbols or nothing is transmitted over the remaining $N_D$ symbol periods, i.e., $N = N_T + N_D$. During the training phase, the signal received by Bob for the $i^{\text{th}}$ channel use is
\begin{equation} \label{eq3}
y_T(i) = \sqrt{P_T}h_{b}x_T(i) + n_b(i)
\end{equation}
where $h_{b}$ is the channel coefficient from Alice to Bob, and $x_T(i)$ is the normalized training signal transmitted by Alice. It is assumed that Bob uses the minimum mean square error (MMSE) technique \cite{kay1993fundamentals} to estimate his channel from Alice.

\subsection{Willie's Detection and Covertness Criterion}
Since Willie is unaware of the slot in which Alice transmits data, he observes all the slots, where in each slot, he makes use of the first $N_T$ pilot symbols to learn the channel coefficient from Alice, and collects the remaining $N_D$ symbols for detection of possible data transmission. In terms of detection, Willie faces a binary hypothesis testing problem where he looks to decide whether Alice transmitted data to Bob or not. We denote Willie's null hypothesis i.e., Alice not transmitting by $\mathcal{H}_0$ while the alternate hypothesis that Alice transmitted a message to Bob is denoted by $\mathcal{H}_1$. Denote $\mathbb{P}_{FA}$ as the probability of false alarm at Willie i.e., Willie decides on $\mathcal{H}_1$ while $\mathcal{H}_0$ is true while $\mathbb{P}_{MD}$ as the probability of missed detection, i.e., Willie decides on $\mathcal{H}_0$ while $\mathcal{H}_1$ is true. It is assumed that Willie has no prior knowledge on whether Alice transmits or not. In such a case, Willie assumes both of the priors to be equally probable\footnote{The effect of assuming a non-trivial prior distribution on Alice's transmission state has been discussed in [6], while a design of un-equal priors has been provided in [16] where a full-duplex receiver of covert information has been used.} which results in Alice achieving covert communication if, for any $\epsilon > 0$, a communication scheme exists such that $\mathbb{P}_{FA} + \mathbb{P}_{MD} \geq 1 - \epsilon$ \cite{bash_jsac,lee_undetect,tamara_jammer_2017}. Here $\epsilon$ signifies the covert requirement, and a sufficiently small $\epsilon$ renders any detector employed at Willie to be ineffective. From Alice and Bob's perspective, it is imperative to force $\epsilon$ to be small, i.e., to achieve strong covertness and pushing Willie into the large detection error regime.

\section{Detection Analysis at Willie}
Willie's detection performance is largely dependent on his knowledge of the channel from Alice, $h_w$. Here, we analyze Willie's detection performance under two extreme cases, i.e., when perfect CSI knowledge is available at Willie and when only CDI is  available. These two cases provide the bounds on Willie's detection performance under the case where he looks to utilize the publicly known pilot symbols transmitted by Alice to learn the channel coefficients. Under the case of perfect CSI knowledge, we first show that Willie's optimal detector is a radiometer, and then proceed to find its optimal detection threshold and the resulting minimum detection error probability. Under the case of unknown CSI at Willie, although the optimal detector is not necessarily a radiometer in general, we will show that Willie's detection performance using a radiometer for unknown CSI converges to his performance under the optimal detector for known CSI, in the large detection error regime. This implies that radiometer is asymptotically the optimal detector in the large detection error regime, regardless of the CSI accuracy at Willie.

\subsection{Detection under Perfect CSI Knowledge}
We consider the scenario when the instantaneous channel realization is perfectly known at Willie. Here, the optimal test that minimizes the detection error at Willie is the likelihood ratio test with $\nu^* = 1$ as the threshold, which is given by
\begin{equation}\label{lrt1}
\frac{\mathbb{P}_1 \triangleq \prod_{i=1}^{N_D} f(y_w(i) | \mathcal{H}_1)}{\mathbb{P}_0 \triangleq \prod_{i=1}^{N_D} f(y_w(i) | \mathcal{H}_0) } \underset{\mathcal{H}_0}{\overset{\mathcal{H}_1}{\gtrless}} 1,
\end{equation}
where $f(y_w(i) | \mathcal{H}_0) = \mathcal{CN}(0, \sigma_w^2)$ and $f(y_w(i) | \mathcal{H}_1) = \mathcal{CN}(0,|h_w|^2 P_D + \sigma_w^2)$ are the likelihood functions of $y_w(i)$ under hypothesis $\mathcal{H}_0$ and $\mathcal{H}_1$, respectively. Here, $y_w(i)$ represents Willie's observation for the $i^{\text{th}}$ symbol duration of the potential data transmission phase, given by
\begin{equation}\label{eq9}
\begin{aligned}
y_w(i) =
\begin{cases}
n_w(i), \quad &\mathcal{H}_0 \\
\sqrt{P_D} h_w  x_D(i) + n_w(i), \quad &\mathcal{H}_1,
\end{cases}
\end{aligned}
\end{equation}
where $x_D$ represents Alice's transmit symbols, with $x_D \sim \mathcal{CN}(0,1)$, and $P_D$ is Alice's data transmit power. It is assumed that Willie is aware of the values used by Alice for $P_D$ and $N_D$ whenever she transmits any data to Bob. Using the distribution of $y_w(i)$, and through performing some algebraic manipulations on (\ref{lrt1}), we have
\begin{equation}
\frac{1}{N_D} \sum_{i=1}^{N_D}|y_w(i)|^2 \underset{\mathcal{H}_0}{\overset{\mathcal{H}_1}{\gtrless}} \lambda^*,
\end{equation}
where $\lambda^*$ is the optimal detection threshold at Willie, and the test statistic given by $T(y_w) = \frac{1}{N_D} \sum_{i=1}^{N_D}|y_w(i)|^2$ shows that the optimal detection test for Willie is to perform a threshold test on the average received power, making a radiometer the optimal detector at Willie under perfect CSI. For the detection error probabilities at Willie, the probability of False Alarm and Missed Detection events is given by
%\begin{equation}
\begin{align}\label{eq10}
\mathbb{P}_{FA} &= \mathbb{P} \left[ \frac{1}{N_D}\sum_{i=1}^{N_D} |y_w(i)|^2 > \lambda | \mathcal{H}_0  \right] \notag \\
&= \mathbb{P} \left[ \chi_{2N_D}^2 > \frac{N_D \lambda}{\sigma_w^2} \right]  = 1 - \frac{\gamma\left(N_D,  \frac{N_D \lambda}{\sigma_w^2} \right)}{\Gamma(N_D)} ,
\end{align}

and

\begin{align}\label{eq11}
\mathbb{P}_{MD} &= \mathbb{P} \left[ \frac{1}{N_D}\sum_{i=1}^{N_D} |y_w(i)|^2 \leq \lambda | \mathcal{H}_1  \right] \notag \\
&=  \mathbb{P}\left[\chi_{2N_D}^2 \leq \frac{N_D \lambda}{|h_w|^2 P_D + \sigma_w^2}   \right]  \notag \\
&= \frac{\gamma\left(N_D,  \frac{N_D \lambda}{|h_w|^2 P_D + \sigma_w^2} \right)}{\Gamma(N_D)} ,
\end{align}

respectively, where $\chi^2_{2 N_D}$ represents a chi-square random variable with $2N_D$ degrees of freedom, $\Gamma(x) = (x-1)!$ is the complete Gamma function, $\gamma(\cdot , \cdot)$ represents the lower incomplete Gamma function, given by
\begin{equation}
\gamma(a,b) = \int_{0}^{b} e^{-x} x^{a-1} dx .
\end{equation}
 The detection error probability at Willie is thus given as

\begin{align}\label{eq12}
\zeta_w &= \mathbb{P}_{FA} + \mathbb{P}_{MD} \notag \\
&= 1 - \frac{\gamma\left(N_D,  \frac{N_D \lambda}{\sigma_w^2} \right)}{\Gamma(N_D)} + \frac{\gamma\left(N_D,  \frac{N_D \lambda}{|h_w|^2 P_D + \sigma_w^2} \right)}{\Gamma(N_D)} .
\end{align}
We next present the optimal choice of the detection threshold of Willie's radiometer and the resulting minimum detection error probability.

\begin{lemma}
Under the assumption of perfect CSI knowledge, the optimal detection threshold of Willie's radiometer for a given channel realization, $h_w$, is
\begin{equation}\label{eq13}
\lambda^*_{CSI} = \frac{\sigma_w^2 (|h_w|^2 P_D + \sigma_w^2)}{|h_w|^2 P_D} \ln \left( \frac{|h_w|^2 P_D + \sigma_w^2}{\sigma_w^2}  \right) ,
\end{equation}
while the resulting minimum detection error probability is given by
%\begin{equation}
\begin{align}\label{eq14}
\zeta_{w, CSI}^* = 1 &- \frac{\gamma\left(N_D, N_D \left(1 + \frac{\sigma_w^2}{|h_w|^2 P_D}\right) \ln (\frac{|h_w|^2 P_D}{\sigma_w^2} + 1) \right)}{\Gamma(N_D)} \notag \\
&+ \frac{\gamma\left(N_D,  \frac{N_D \sigma_w^2 }{|h_w|^2 P_D}  \ln (\frac{|h_w|^2 P_D}{\sigma_w^2} + 1) \right)}{\Gamma(N_D)} .
\end{align}
\end{lemma}

\begin{proof}
To minimize the detection error probability, Willie considers the problem:
\begin{equation}\label{eq15}
\underset{\lambda}{\mathrm{min}} \quad \zeta_w = \mathbb{P}_{FA} + \mathbb{P}_{MD}.
\end{equation}
From the definition of upper and lower incomplete Gamma functions, $\Gamma(s) = \Gamma(s,x) + \gamma(s,x)$, where $\Gamma(\cdot, \cdot)$ is the corresponding upper incomplete Gamma function. Thus, we can write
%\begin{equation}
\begin{align}\label{eq16}
\zeta_w = 1 -  \frac{1}{\Gamma(N_D)} &\Bigg[ \Gamma\left(N_D,  \frac{N_D \lambda}{|h_w|^2 P_D + \sigma_w^2} \right) \notag \\
&- \Gamma\left(N_D,  \frac{N_D \lambda}{\sigma_w^2} \right) \Bigg]   .
 \end{align}
%\end{equation}
Setting $\frac{\partial \zeta_w}{\partial \lambda} = 0$ and some algebraic manipulations give the optimal value of $\lambda$, where we use derivative property of the upper incomplete Gamma function, given by:
\begin{equation}\label{eq17}
\frac{\partial \Gamma(s, f(x))}{\partial x} = - (f(x))^{s-1} e^{-f(x)} \frac{\partial f(x)}{\partial x} .
\end{equation}
Next, putting in the value of $\lambda^*_{CSI}$ into the expression for $\zeta_w$ in (\ref{eq12}) gives the desired result for $\zeta_{w, CSI}^*$.
\end{proof}

\subsection{Detection under Knowledge of CDI only}
\noindent
In this subsection, we consider the scenario where Willie does not know the channel coefficient, and only the channel distribution information is available to Willie. Under the assumption of a radiometer, the detection error probability at Willie still has the same expression as given in (\ref{eq12}). However, since Willie is unaware of his instantaneous channel realizations from Alice, the optimal detection threshold at Willie in this case is given by
\begin{equation}\label{eq19}
\lambda_{CDI}^* = \underset{\lambda}{\mathrm{ \arg \min}} \: \mathbb{E}_{|h_w|^2} \left[ \zeta_{w, CDI} \right] ,
\end{equation}
where the expectation is taken over the distribution of $h_w$.

\subsection{Performance Comparison between CSI and CDI Cases}
From Alice and Bob's perspective, achieving strong covertness implies having large detection errors at Willie which, in turn, requires Alice to transmit at very low powers. Here, we show that for these low transmit power level transmissions, as it is generally considered in covert communication scenarios, the optimal detection performance at Willie under the perfect CSI case and CDI only case are indistinguishable. To show this, we first present linear approximations of Willie's detection error probability in the asymptotically low power regime ( i.e., around $P_D \rightarrow 0$ ) under perfect CSI and CDI only cases, and then establish the equivalence of these linear approximations.

\begin{lemma}
The linear approximation of $\zeta_{w, CSI}^*$ for a given channel realization in the asymptotically low power regime, is given as
\begin{equation}\label{eq20}
 \underset{P_D \rightarrow 0}{\mathrm{\lim}} \: \zeta_{w, CSI}^* = 1 - \frac{|h_w|^2 N_D^{N_D} e^{-N_D}}{\sigma_w^2 \Gamma(N_D)} P_D + o(P_D) ,
\end{equation}
where $o(P_D)$ represents the remainder terms of the series expansion.
\end{lemma}

\begin{proof}
See Appendix A.
\end{proof}

We next present a linear approximation for $\zeta_{w, CDI}^*$, which is Willie's optimal detection error probability under the case where only CDI is available to Willie.
\begin{lemma}
The linear approximation of $\zeta_{w, CDI}^*$ for a given channel realization, in the asymptotically low power regime, is given as
\begin{equation}\label{eq35}
 \underset{P_D \rightarrow 0}{\mathrm{\lim}} \: \zeta_{w, CDI}^* = 1 - \frac{|h_w|^2 N_D^{N_D} e^{-N_D}}{\sigma_w^2 \Gamma(N_D)} P_D + o(P_D) ,
\end{equation}
where $o(P_D)$ represents the remainder terms of the series expansion.
\end{lemma}
\begin{proof}
See Appendix B.
\end{proof}

\begin{prop}
For a given channel realization, the linear approximation of Willie's detection error probability under perfect CSI, $\zeta_{w, CSI}^*$, and under CDI only, $\zeta_{w, CDI}^*$, are equivalent in the asymptotically low power regime.
\end{prop}
\begin{proof}
Results of Lemma 2 and Lemma 3 provide the desired equivalence.
\end{proof}
For a further insight into the results given above, a couple of remarks are in order.

\begin{remark}
From Proposition $1$, Willie's optimal (minimum) detection error probabilities under the cases of perfect CSI and CDI only are asymptotically indistinguishable in the large detection error regime. The numerical validation of this equivalence is provided in Fig. 2. This equivalence implies that the accuracy of CSI at Willie does not change his detection performance that much as long as Willie's detection error probability is forced to be close to $1$. From Alice and Bob's perspective, they are unaware of the CSI's accuracy at Willie and want to ensure large detection errors. Therefore, we use $\zeta_{w, CSI}^*$ as the detection error probability at Willie under training. Although this constitutes a worst case scenario from the perspective of covert communication design, it does yield a more robust, yet accurate, approach.
\end{remark}

\begin{remark}
We note that the equivalence result obtained in Proposition 1 is based on a radiometer as the detector. It has been shown earlier that under the case of perfect CSI, radiometer is indeed the optimal detector. However, it is not clear that whether it is also the case under CDI only. Proposition 1 tells that, in the large detection error regime, Willie's detection with CDI using radiometer (which may not be optimal in general) already gives almost the same performance as detection with perfect CSI using radiometer (which is optimal). This implies that the radiometer is asymptotically the optimal detector with any accuracy of CSI, ranging from CDI only to perfect CSI, in the large detection error regime.
\end{remark}

\section{Covertness under Channel Uncertainty}
In this section, we first describe the channel estimation at Bob. Next, we consider a system metric that affects the covert communication performance, and then find the optimal solution to the covertness problem at Alice. We allow Alice to choose different power levels for pilot and data transmission. For simplicity, the training duration is fixed to one symbol which is in agreement with previous works on training-based communications \cite{cavers1991analysis,he2013secure}. It should be noted here that a higher number of training symbols will only result in a better channel estimate at Bob, improving the covert system performance, as Willie's detection is already considered under the case of perfect CSI. In addition, the power of the pilot symbol is set to the maximum allowable transmit power, i.e., $P_T = P_{\text{max}}$, Under this setup, the problem at Alice is of finding the optimal power for data transmission and the number of symbols used for data in order to maximize the covert throughput under a given covertness constraint. We note here that from Alice and Bob's perspective, it is desirable to keep Willie in the large detection error regime for achieving strong covertness.

\subsection{Channel Estimation at Bob}
As mentioned in Section II, Bob applies the LMMSE technique to estimate the channel from Alice during the training phase. The estimation of channel coefficient and corresponding estimation error at Bob are denoted by $\hat{h}_{b}$ and $\tilde{h}_{b}$, respectively. Thus
\begin{equation} \label{eq4}
h_{b} = \hat{h}_{b} + \tilde{h}_{b} ,
\end{equation}
where $\hat{h}_{b}$ and $\tilde{h}_{b}$ follow zero mean CSCG distributions \cite{gursoy2009capacity}. Furthermore, since $y_T$ is a linear function of the channel coefficient, the linear MMSE estimation becomes the optimal MMSE estimation, and the orthogonality principle implies that $\mathbb{E}\left[|h_{b}|^2  \right] = \mathbb{E}[|\hat{h}_{b}|^2 ] + \mathbb{E}[|\tilde{h}_{b}|^2 ]$. Based on LMMSE, the estimate of $h_b$ is given by \cite{gursoy2009capacity}
\begin{equation}\label{eq5}
\hat{h}_b = \frac{\sqrt{P_T}}{\sigma_b^2 + N_T P_T}y_T x_T^{\dagger} .
\end{equation}
We define $\beta_b$ as the variance of the channel estimation error at Bob, i.e, $\beta_b = \mathbb{E}[|\tilde{h}_{b}|^2 ]$, and resultantly, $\mathbb{E}[|\hat{h}_{b}|^2 ] = 1 - \beta_b$, where \cite{vakili2006effect}
\begin{equation}\label{eq7}
\beta_b = \frac{\sigma_b^2}{\sigma_b^2+ N_T P_T}.
\end{equation}
Since Bob is aware of the slot in which Alice transmits, he performs channel estimation only in such a slot and then uses the obtained channel estimate to perform data detection in the corresponding transmission slot.

\subsection{Covert Connection Probability}
During the covert data transmission, Alice transmits at a fixed, pre-determined rate\footnote{We assume here that we have a fixed rate transmitter without a degree of freedom to change $R$. Furthermore, we would like to emphasize that making $R$ a design parameter would not change the main conclusion of this work. This can be attributed to the fact that Willie's detection performance depends on Alice's data transmit power and the number of data symbols used, but does not directly depend on the data rate.} to Bob which is denoted by $R$. Due to the random nature of quasi-static fading channels from Alice to Bob, a transmission outage occurs whenever $C \leq R$, where $C$ is the capacity of the Alice to Bob channel, and in case of a transmission outage, Bob is unable to reliably decode the information transmitted by Alice. We note that for quasi-static fading channels under finite blocklength, the channel dispersion associated with finite blocklength approximation is zero \cite{yang2013block,yang2013quasi,yang2014tit}, and channel fading becomes the dominant source of decoding error events, resulting in channel outage. Under such a consideration, the outage probability can be used to accurately describe the communication performance \cite{durisi2016toward}. Here, we derive the complement of outage probability, defined as the covert connection probability, which is the probability that Bob can reliably decode a covert message from Alice, transmitted at a fixed rate $R$. The covert connection probability, $P_{cc}$, is

\begin{equation}\label{eq47}
P_{cc} = 1 - \mathbb{P} \left[ \log_2 (1 + \gamma_b) \leq R  \right]
\end{equation}
where $\gamma_b$ denotes the signal-to-noise ratio at Bob which, under the considered channel uncertainty model, is given by \cite{he2013secure}
\begin{equation}\label{eq48}
\gamma_b = \frac{|\hat{h}_b|^2 P_D}{|\tilde{h}_b|^2 P_D + \sigma_b^2}.
\end{equation}
In the following, we present the expression for the desired covert connection probability.

\begin{lemma}
The covert connection probability for Alice to Bob transmission at a fixed rate $R$, and under channel uncertainty at Bob, is given by
\begin{equation}\label{eq49}
P_{cc} = \frac{1-\beta_b}{(1-\beta_b) + \beta_b (2^R - 1)} e^{- \frac{\sigma_b^2(2^R - 1)}{(1-\beta_b)P_D}},
\end{equation}
where $P_D$ is Alice's transmit power during data transmission and $\beta_b$ is the variance of channel estimation error at Bob, as defined in (\ref{eq7}).
\end{lemma}

\begin{proof}
Putting in the expression for $\gamma_b$ into the expression of $P_{cc}$, we have
%\begin{equation}
\begin{align}\label{eq50}
P_{cc} &= 1 - \mathbb{P}\left[\log_2 (1 + \frac{|\hat{h}_b|^2 P_D}{|\tilde{h}_b|^2 P_D + \sigma_b^2} ) \leq R   \right] \notag \\
&= 1 - \mathbb{P} \left[|\hat{h}_b|^2 \leq \frac{(2^R-1)(|\tilde{h}_b|^2 P_D + \sigma_b^2)}{P_D}  \right] .
\end{align}
%\end{equation}
Then using the distribution of $|\hat{h}_b|^2$ and $|\tilde{h}_b|^2$ gives
\begin{align}
P_{cc} &= \frac{1}{\beta_b} \int_{0}^{\infty} e^{- \frac{(2^R-1)(|\tilde{h}_b|^2 P_D + \sigma_b^2)}{(1-\beta_b)P_D}-\frac{|\tilde{h}_b|^2}{\beta_b}} \mathrm{d}|\tilde{h}_b|^2  \notag \\
&= \frac{1}{\beta_b} e^{ - \frac{(2^R-1) \sigma_b^2}{(1-\beta_b)P_D}} \int_{0}^{\infty} e^ {- |\tilde{h}_b|^2  \frac{ \left((1-\beta_b)P_D + \beta_b P_D (2^R-1)\right)}{\beta_b (1-\beta_b)P_D}}  \mathrm{d}|\tilde{h}_b|^2 \notag \\
&= \frac{(1-\beta_b)P_D}{(1-\beta_b)P_D + \beta_b P_D (2^R-1)} e^{ - \frac{\sigma_b^2(2^R-1)}{(1-\beta_b)P_D}},
\end{align}
which concludes the proof.
\end{proof}

\subsection{Optimal Transmit Power and Number Of Transmit Symbols}
As discussed in Remark $1$, we consider $\zeta_{w, CSI}^*$ provided in (\ref{eq14}) as the minimum detection error probability at Willie under channel uncertainty, simply denoting it by $\zeta_{w}^*$. Since Alice is unaware of her channel realization to Willie, she considers the expected value of $\zeta_{w}^*$ over all possible realizations of her channel to Willie as the detection metric. Here, Alice looks to maximize her covert throughput to Bob while ensuring that Willie's average detection error probability satisfies a given covertness constraint. Owing to delay requirements, we assume in this work that the transmitted signals are constrained by a maximum blocklength, $N_{\text{D,max}}$, thus the number of Alice's covert data symbols is limited by $N_D \leq N_{\text{D,max}}$. On the other hand, there also exists a limit on the minimum number of symbols Alice can use due to the channel coding requirements for short-packet communications \cite{durisi2016toward,maiya2012low}, and this limit is denoted by $N_{\text{D,min}}$. This requirement on the minimum number of symbols is also needed for the outage-based approach to hold \cite{durisi2016toward}. In regards to the transmit power, a maximum transmit power constraint at Alice is considered, given by $P_{\text{max}}$. As mentioned previously, Alice uses the maximum allowed transmit power, $P_{\text{max}}$, for the pilot symbol.

The design problem at Alice is to optimally choose the data transmission power and the number of data symbols for covert communication, stated as
\begin{subequations}\label{eq51}
\begin{align}
\textbf{P1} \quad \underset{P_D, N_D}{\mathrm{maximize}} \quad & N_D R P_{cc} \notag \\
\mathrm{subject\:to}\quad &\mathbb{E}_{|h_{w}|^2}\left[ \zeta_{w}^* \right] \geq 1 - \epsilon \\
&P_D \leq P_{\text{max}}\\
&N_{\text{D,min}} \leq N_D \leq N_{\text{D,max}},
\end{align}
\end{subequations}
where $N_D R P_{cc}$ is the throughput from Alice to Bob, and the design parameters $P_D$ and $N_D$ refer to Alice's data transmission power and the number of symbols used for data transmission, respectively. Here, $\epsilon$ signifies the desired level of covertness, whereas $\zeta_{w}^*$ is as given in (\ref{eq14}), and in the statement of \textbf{P1}, $P_T = P_{\text{max}}$ is assumed. The solution to this problem is stated in the following.

\begin{lemma}
Alice's optimal transmit power for data transmission, as a function of $N_D$, is given by
\begin{equation}\label{eq53}
\begin{aligned}
P_D^* =
\begin{cases}
P_D^{\dagger}(N_D),\quad &\text{If} \quad P_D^{\dagger}(N_D) \leq P_{\text{max}}\\
P_{\text{max}}, &\text{Otherwise},
 \end{cases}
 \end{aligned}
\end{equation}
where $P_D^{\dagger}(N_D)$ is the solution to $\mathbb{E}_{|h_{w}|^2}\left[ \zeta_{w}^* \right] = 1 - \epsilon$ for a given $N_D$. The optimal number of data symbols transmitted by Alice is given by
\begin{equation}\label{eq54}
\begin{aligned}
N_D^* =
\begin{cases}
N_{\text{D,min}},\quad &\text{If} \quad N_D^{\dagger} \leq N_{\text{D,min}}\\
N_D^{\dagger}, &\text{If} \quad N_{\text{D,min}} < N_D^{\dagger} \leq N_{\text{D,max}}\\
N_{\text{D,max}}, &\text{Otherwise},
 \end{cases}
 \end{aligned}
\end{equation}
where $N_D^{\dagger}$ is the solution for $N_D$ to
\begin{equation}\label{eq55}
\underset{N_D}{\mathrm{maximize}} \quad N_D R P_{cc},
\end{equation}
and $P_{cc}$ is a function of $N_D$ in terms of $P_D$.
\end{lemma}

\begin{proof}
We first note that for a fixed $P_T = P_{max}$, the covert connection probability, $P_{cc}$, is an increasing function of $P_D$. On the other hand,  $\mathbb{E}_{|h_{w}|^2}\left[ \zeta_{w}^* \right]$ is a decreasing function of $P_D$, hence a given solution will satisfy the constraint at equality. From the constraint at equality and a given $N_D$, the solution for $P_D$, as indicated by $P_D^{\dagger}(N_D)$, can be obtained. This results in the one-dimensional optimization problem in (\ref{eq55}), which can be solved by performing a numerical search over all possible values of $N_D$. Incorporating the maximum and minimum value of $P_D$ and $N_D$ gives the desired result.
\end{proof}

We note that the optimal solution presented in Lemma 5 does not yield a closed form expression for $P_D^*$ and $N_D^*$. Rather, the solution relies on numerical search methods \cite{chapra2010numerical} to solve the optimization problem in (\ref{eq55}). We next present a suboptimal closed-form solution to this problem.

\subsection{Suboptimal Solution}
Based on the linear approximation in the asymptotically low power regime (small $\epsilon$ regime) developed earlier, we present here a suboptimal solution to find closed form expressions for $P_D^*$ and $N_D^*$. Using the linear approximation for $\zeta_{w}^*$, we rewrite the problem at Alice as
\begin{subequations}\label{eq56}
\begin{align}
\textbf{P1.1} \quad \underset{P_D, N_D}{\mathrm{maximize}} \quad &N_D R P_{cc} \notag \\
\mathrm{subject\:to}\quad &\mathbb{E}_{|h_{w}|^2}\left[ \zeta_{w}^* \right] \geq 1 - \epsilon \\
&P_D \leq P_{max}\\
&N_{D, min} \leq N_D \leq N_{D, max},
\end{align}
\end{subequations}
where now,
\begin{equation}\label{eq57}
\zeta_{w}^* \approx 1 - \frac{|h_w|^2 N_D^{N_D} e^{-N_D}}{\sigma_w^2 \Gamma(N_D)} P_D.
\end{equation}
The solution to this problem is presented in the following.

\begin{lemma}
In the asymptotically small $\epsilon$ regime, Alice's optimal transmit power for data transmission is given by
\begin{equation}\label{eq58}
\begin{aligned}
P_D^* =
\begin{cases}
P_D^{\ddagger},\quad &\text{If} \quad P_D^{\ddagger} \leq P_{max}\\
P_{max}, &\text{Otherwise},
\end{cases}
\end{aligned}
\end{equation}
where
\begin{equation}\label{eq59}
P_D^{\ddagger} = \frac{\epsilon \sigma_w^2 \Gamma(N_D^*)}{(N_D^*)^{N_D^*} e^{-N_D^*}},
\end{equation}
and the optimal number of data symbols transmitted by Alice is $N_D^* = N_{D, min}$.
\end{lemma}

\begin{proof}
Under the exponential distribution of $|h_w|^2$, the expectation is calculated as
\begin{equation}\label{eq60}
\mathbb{E}_{|h_{w}|^2}\left[ \zeta_{w}^* \right] = 1 - \frac{N_D^{N_D} e^{-N_D}}{\sigma_w^2 \Gamma(N)} P_D,
\end{equation}
and the covertness constraint then gives
\begin{equation}\label{eq61}
P_D \leq \frac{\epsilon \sigma_w^2 \Gamma(N_D)}{(N_D)^{N_D} e^{-N_D}}.
\end{equation}
We note that $P_{cc}$ is an increasing function of $P_D$ while the covertness constraint puts an upper bound on $P_D$, hence a given solution will satisfy the constraint at equality. This results in the optimization problem given as
\begin{equation}\label{eq62}
\underset{N_D}{\mathrm{maximize}} \quad N_D R P_{cc},
\end{equation}
where $P_{cc}$ is now a function of $N_D$. Considering the partial derivative w.r.t. $N_D$, we have

\begin{flalign}\label{eq63}
&\frac{\partial (N_D R P_{cc})}{\partial N_D}   \notag \\
= &- \frac{(1-\beta_b)R}{(1-\beta_b) + \beta_b (2^R - 1)} \cdot e^{- \frac{e^{-N_D}\left( N_D e^{N_D} \Gamma(N_D) + A N_D^{N_D} \right)}{\Gamma(N_D)}} \notag \\
&\times \left[ \frac{A N_D^{N_D+1}\left(\ln(N_D) - \psi(N_D)   \right) - e^{N_D} \Gamma(N_D) }{\Gamma(N_D)} \right] ,
\end{flalign}

which is strictly negative for $N_D \geq 1$. Here $A = \frac{\sigma_b^2 (2^R - 1)}{\sigma_w^2 (1 - \beta_b) \epsilon}$ and $\psi(x)$ is the Digamma function, which is defined as $\psi(x) = \frac{\Gamma'(x)}{\Gamma(x)}$. Thus the value of $N_D$ maximizing the throughput is the minimum allowed $N_D$, i.e., $N_{D, min}$. This concludes the proof.
\end{proof}

\section{Numerical Results and Discussions}

In this section, we present the numerical results and study the performance of the considered covert communication scenario under given covertness constraints. Unless stated otherwise, we consider a pre-determined rate for Alice to Bob transmission of $R=1$, the variance of Willie's receiver noise is set to $\sigma_w^2 = 0.05$, while the variance of Bob's receiver noise is set to $\sigma_b^2 = 0.01$. We consider a maximum power constraint of $P_{\text{max}} = 1$ at Alice, while $N_{D, \text{min}}$ and $N_{D, \text{max}}$ are set to be 50 and 100, respectively. We note that in the literature related to short packet communication, blocklengths in the range of 50-200 have been used \cite{durisi2016toward,mary2016finite,ren2019resource}, while for practical error correcting codes, blocklengths of $n=128, 256$ and 512 have been shown to perform well in the desired decoding error probability range \cite{shirvanimoghaddam2018short}.

We first provide a numerical validation for the equivalence of Willie's detection error probability under the cases of perfect CSI and CDI only in the large detection error regime, as derived in Proposition 3, and also explained in Remark 1. In Fig. \ref{fig_sim_2}, we plot these detection error probabilities at Willie against a range of Alice's data transmit power, $P_D$, for different numbers of data transmit symbols, $N_D$. We first note that as $N_D$ or $P_D$ increases, Willie's detection performance improves. More importantly, Willie's detection performances are indistinguishable between the perfect CSI case and the CDI only case in the large detection error regime, e.g., $\zeta_w^* \geq 0.9$. The detection performances of the two cases are still very close to each other even at $\zeta_w^* = 0.8$. These results validate our analysis and the conclusion that Willie's detection performance is extremely insensitive to the CSI's accuracy as long as the detection error probability is forced by Alice and Bob to be fairly close to 1.

\begin{figure}[t!]
\centering
	\includegraphics[width=\linewidth]{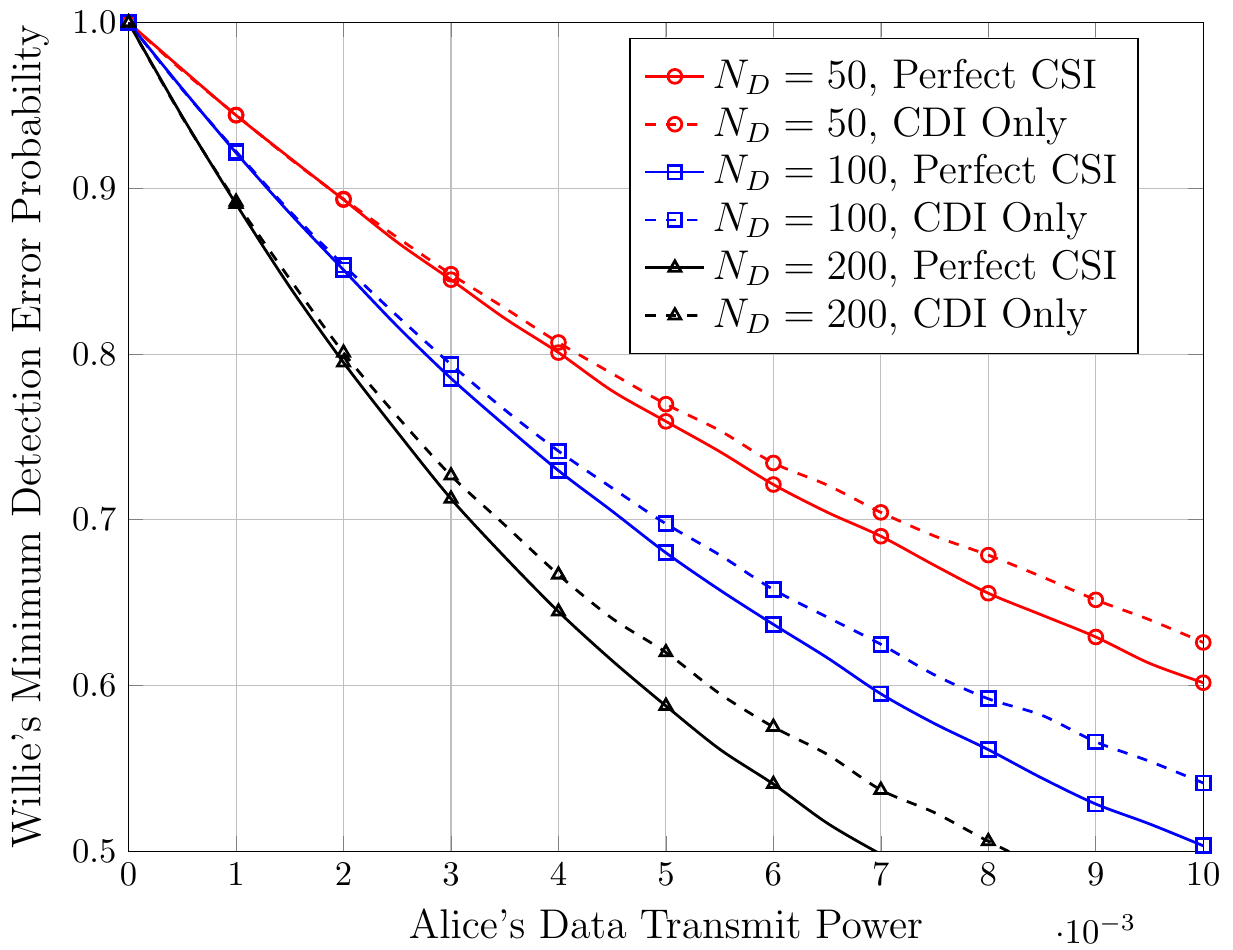}
	\caption{Willie's minimum detection error probability, $\zeta_w^*$, vs. Alice's data transmit power, $P_D$, under perfect CSI and CDI only cases for varying $N_D$.}
	\label{fig_sim_2}
\end{figure}

\begin{figure}[t!]
\centering
	\includegraphics[width=\linewidth]{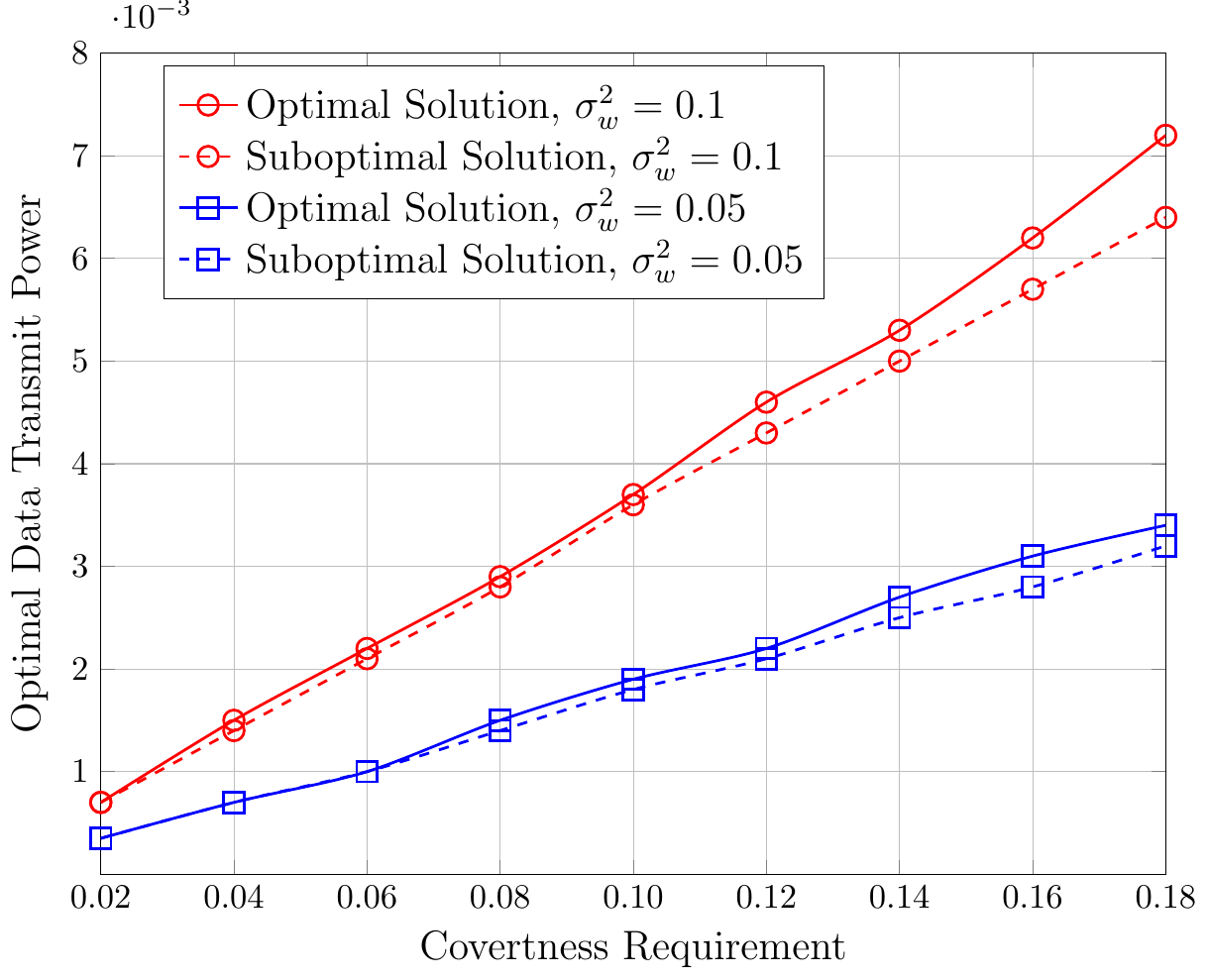}
	\caption{Comparison of the optimal data transmit power at Alice, $P_D^*$, under the optimal and suboptimal solution vs. the covertness requirement, $\epsilon$.}
\label{fig_sim_3}
\end{figure}

We next present the optimal choice of Alice's data transmit power and the optimal number of data transmit symbols under given covertness constraints in Fig. \ref{fig_sim_3} and Fig. \ref{fig_sim_4}, respectively, where we also plot the best choice for Alice's parameters under the proposed suboptimal scheme. We show these results for two different sets of noise variances at Willie for ease of exposition. Firstly, for the optimal data power values, we see that since a higher noise power causes an increased uncertainty in Willie's observations, Alice can transmit to Bob using a higher transmit power. Secondly, the proposed suboptimal scheme performs very close to the optimal one, especially in the low transmit power regime. We also note here that since the proposed suboptimal scheme is based on the linear approximation of Willie's detection performance around $P_D \rightarrow 0$, the curves for optimal transmit power deviate further from each other as the covertness requirement is relaxed, resulting in Willie no longer operating in the large detection error regime. Regarding the optimal number of data transmit symbols at Alice, both the optimal and suboptimal scheme provide the same solution, i.e., to use the minimum possible number of transmit symbols, $N_{\text{D,min}}$. We would like to emphasize here that $N_{D, \text{min}}$, which comes to be the optimal choice for $N_D$, can not be made arbitrarily small, owing to its relation to channel coding constraints and to the adopted outage-based approach.

\begin{figure}[t!]
\centering
	\includegraphics[width=\linewidth]{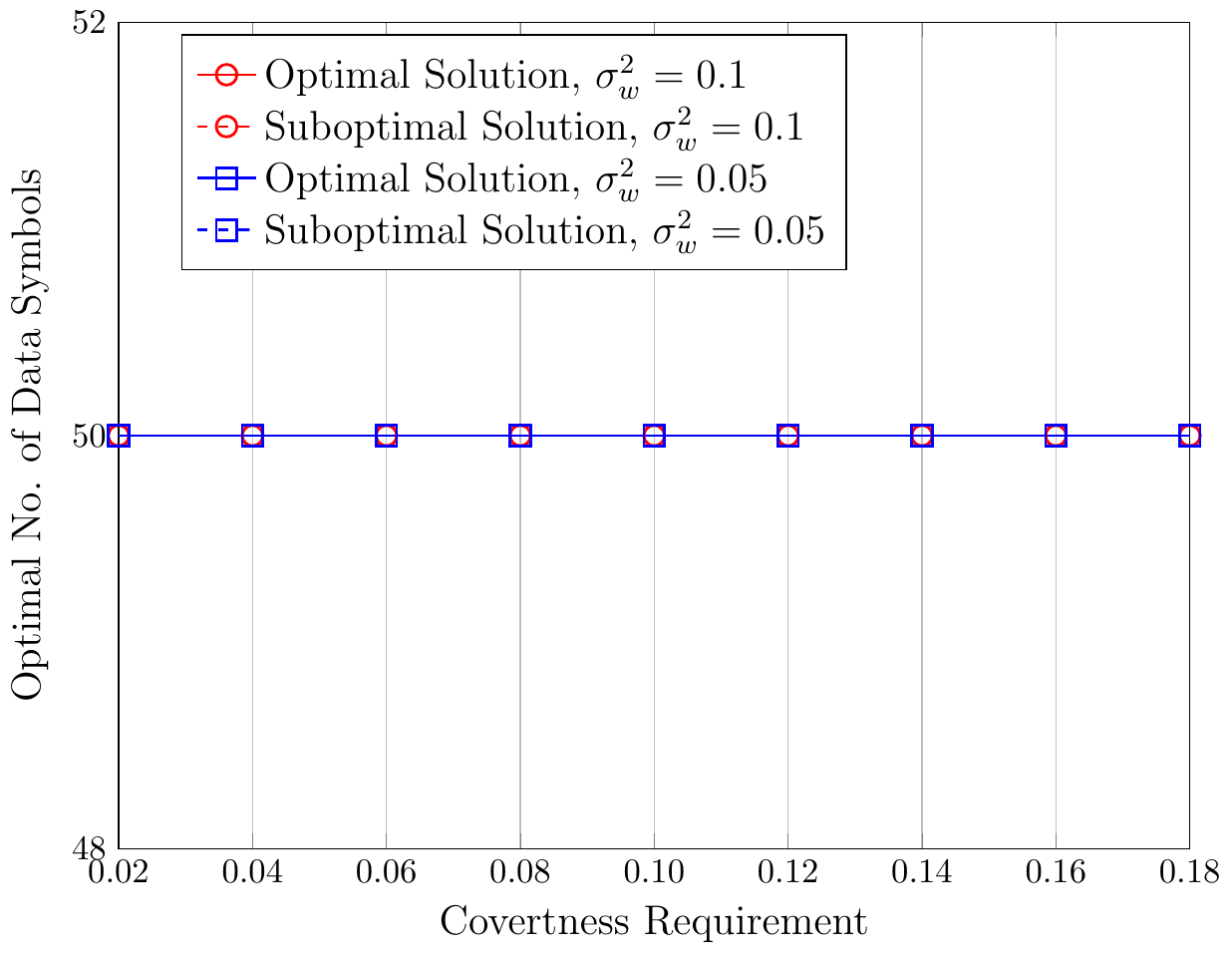}
	\caption{The optimal number of data symbols used by Alice, $N_D^*$, under the optimal and suboptimal solution vs. the covertness requirement, $\epsilon$. Note that all four curves in this figure overlap completely.}
\label{fig_sim_4}
\end{figure}

\begin{figure}[t!]
\centering
	\includegraphics[width=\linewidth]{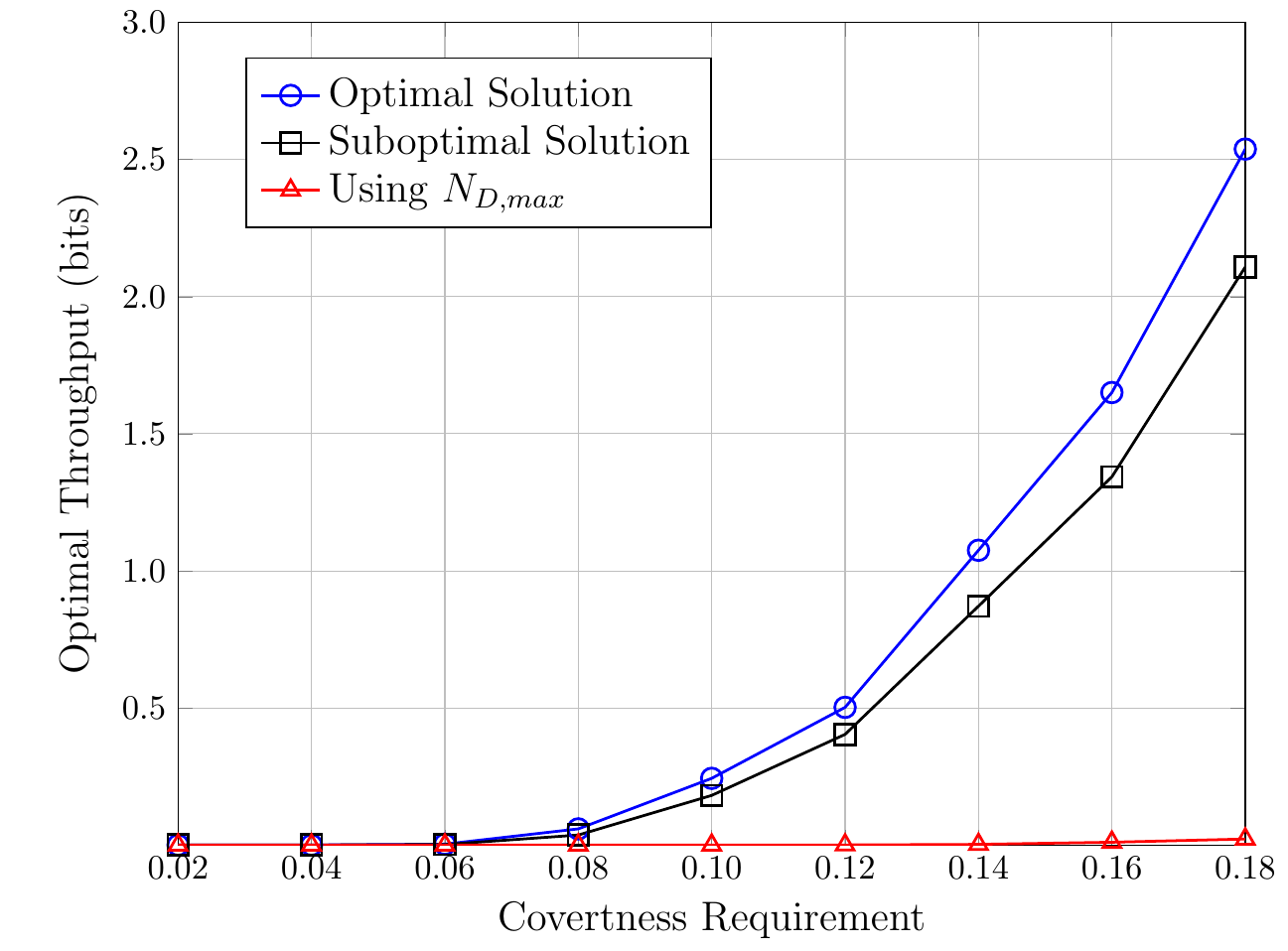}
	\caption{The optimal throughput from Alice to Bob, $N_D^* R P_{cc}$, under the optimal approach, suboptimal approach and case of using $N_{D,\text{max}}$ vs. the covertness requirement, $\epsilon$.}
	\label{fig_sim_5}
\end{figure}

It is important to highlight that the optimal (and suboptimal) solution of only using the \textit{minimum} number of transmit symbols is in sharp contrast to the previously established result for non-fading AWGN channels in \cite{yan2019delay}. Specifically, it was shown in \cite{yan2019delay} that it is optimal to use the \textit{maximum} allowable number of transmit symbols to maximize the covert throughput. This comparison demonstrates a fundamental difference in the covert transmission design between the non-fading AWGN channel and the quasi-static fading channel. Whereas in the non-fading AWGN channel case, the decoding errors are caused by the finite blocklength (i.e., the number of symbols used to transmit the message), in our case, the effect of channel fading is another reason for causing decoding errors. This effect of channel fading has been shown to be the dominant effect in causing decoding errors \cite{yang2013block,yang2013quasi,yang2014tit}, while the impact of blocklength becomes negligible under quasi-static fading. To further illustrate the importance of appropriate design, we investigate the advantage of using the optimal and suboptimal solutions over a scheme where Alice uses the maximum allowable number of symbols in a communication slot (with optimized data transmit power). Fig. \ref{fig_sim_5} shows the covert throughput achieved under the optimal and suboptimal solutions with $N_D^* = N_{\text{D,min}}$, and the covert throughput achieved by using $N_D = N_{\text{D,max}}$. The difference between the optimal and suboptimal schemes can again be attributed to the fact that Willie is no longer operating in the large detection error regime, causing a deviation in the results obtained through the optimal approach and suboptimal approach based on the linear approximation around $P_D \rightarrow 0$. We note a significant difference in the achieved throughput between the use of $N_{\text{D,min}}$ and $N_{\text{D,max}}$. Specifically, we see that the optimal (and suboptimal) solution achieves $110$-fold more throughput than that achieved by using the maximum number of data symbols. We also observe that the throughput of the suboptimal solution is roughly $20\%$ lower than that of the optimal solution, due to the small but non-negligible difference in the transmit power designs. Hence, such a moderate performance reduction is the price to pay for using the closed-form suboptimal design with minimum complexity.

\section{Conclusion}
In this paper, we have considered covert communications under the scenario where users suffer from channel uncertainty while Alice uses pilot symbols to help the intended receiver estimate their channel. We have derived the optimal detection threshold for Willie and the resulting minimum detection error probability under the extreme cases of the availability of complete CSI and CDI only at Willie. It has been shown that in the low transmit power regime, the two extreme cases are indistinguishable and hence, the quality of channel knowledge at Willie does not improve his detection performance as long as it is forced to stay in the large detection error regime. From the covert communications pair perspective, we provide the optimal choice for data transmit power and the optimal number of data transmit symbols that maximize the covert throughput. As long as there is a sufficient number of data symbols for the outage-based analysis to hold, we find that using a smaller number of data symbols achieves a higher covert throughput.

\section*{Appendix A \\ Proof of Lemma 2}
We note that for $P_D=0$, the expression of $\zeta_{w, CSI}^*$ gives a value of 1. This is expected since  in case of no transmission by Alice, Willie is unable to distinguish between the two hypotheses. This value also serves as the intercept of the linear (first order) approximation of $\zeta_{w, CSI}^*$ as a function of $P_D$. To complete the approximation, we need to find the slope of $\zeta_{w, CSI}^*$ as $P_D \rightarrow 0$, i.e, $\underset{P_D \rightarrow 0}{\mathrm{\lim}} \: \frac{\partial \zeta_{w, CSI}^*}{\partial P}$. Using the relationship of upper and lower incomplete Gamma functions given by $\Gamma(a) = \Gamma(a,b) + \gamma(a,b)$, we have

\begin{align}\label{eq21}
\zeta_{w, CSI}^* = 1 &- \frac{\Gamma\left(N_D,  \frac{N_D \sigma_w^2 }{|h_w|^2 P_D}  \ln (\frac{|h_w|^2 P_D}{\sigma_w^2} + 1) \right)}{\Gamma(N_D)} \notag \\
&+ \frac{\Gamma\left(N_D, N_D\left(1 + \frac{\sigma_w^2}{|h_w|^2 P_D}\right) \ln (\frac{|h_w|^2 P_D}{\sigma_w^2} + 1) \right)}{\Gamma(N_D)} .
\end{align}
To calculate the desired derivative, we consider the terms in $\zeta_{w, CSI}^*$ separately, where we rely on the derivative property of upper incomplete Gamma function, given by
\begin{equation}\label{eq22}
\frac{\partial \Gamma \left(s,f(x) \right)}{\partial x} = - \left( f(x) \right)^{s-1} e^{-f(x)} \frac{\partial f(x)}{\partial x}.
\end{equation}
The derivative for the second term of $\zeta_{w, CSI}^*$ in (\ref{eq21}) is calculated as

\begin{align}\label{eq23}
\frac{1}{\Gamma(N_D)}&\frac{\partial \Gamma\left(N_D,  \frac{N_D \sigma_w^2 }{|h_w|^2 P_D}  \ln (\frac{|h_w|^2 P_D}{\sigma_w^2} + 1) \right)}{\partial P_D}\notag \\
= &- \frac{1}{\Gamma(N_D)} \left[\frac{N_D \sigma_w^2 }{|h_w|^2 P_D}  \ln \left(\frac{|h_w|^2 P_D}{\sigma_w^2} + 1\right) \right]^{N_D-1}\notag \\
& \times e^{- \frac{N_D \sigma_w^2 }{|h_w|^2 P_D}  \ln \left(\frac{|h_w|^2 P_D}{\sigma_w^2} + 1\right) } \notag \\
& \times \frac{\partial}{\partial P_D}\left[ \frac{N_D \sigma_w^2 }{|h_w|^2 P_D}  \ln \left(\frac{|h_w|^2 P_D}{\sigma_w^2} + 1\right) \right] \notag \\
= &- \frac{N_D^{N_D}}{\Gamma(N_D)} \left[ \frac{\sigma_w^2 }{|h_w|^2 P_D}  \ln \left(\frac{|h_w|^2 P_D}{\sigma_w^2} + 1\right) \right]^{N_D-1} \notag \\
& \times e^{- \frac{N_D \sigma_w^2 }{|h_w|^2 P_D}  \ln \left(\frac{|h_w|^2 P_D}{\sigma_w^2} + 1\right) } \notag \\
& \times \left[\frac{\sigma_w^2}{P_D (|h_w|^2 P_D + \sigma_w^2)} - \frac{\sigma_w^2 \ln \left( \frac{|h_w|^2 P_D}{\sigma_w^2} + 1 \right)}{|h_w|^2 P_D^2} \right] .
\end{align}

Similarly, the derivative for the third term of $\zeta_{w, CSI}^*$ in (\ref{eq21}) is calculated as

\begin{align}\label{eq24}
&\frac{1}{\Gamma(N_D)} \frac{\partial \Gamma\left(N_D, N_D\left(1 + \frac{\sigma_w^2}{|h_w|^2 P_D}\right)  \ln (\frac{|h_w|^2 P_D}{\sigma_w^2} + 1) \right)}{\partial P_D} \notag \\
= &- \frac{1}{\Gamma(N_D)} \Bigg[ N_D \left(1 + \frac{\sigma_w^2}{|h_w|^2 P_D}\right) \ln \left( \frac{|h_w|^2 P_D}{\sigma_w^2} + 1 \right)  \Bigg]^{N_D-1} \notag \\
& \times e^{- N_D \left(1 + \frac{\sigma_w^2}{|h_w|^2 P_D}\right) \ln \left( \frac{|h_w|^2 P_D}{\sigma_w^2} + 1 \right) } \notag \\
& \times \frac{\partial}{\partial P_D} \left[ N_D \left(1 + \frac{\sigma_w^2}{|h_w|^2 P_D}\right) \ln \left( \frac{|h_w|^2 P_D}{\sigma_w^2} + 1 \right)  \notag \right]
\end{align}

\begin{flalign}
= &- \frac{N_D^{N_D}}{\Gamma(N_D)} \Bigg[ \left(1 + \frac{\sigma_w^2}{|h_w|^2 P_D}\right) \ln \left( \frac{|h_w|^2 P_D}{\sigma_w^2} + 1 \right)  \Bigg]^{N_D-1} \notag \\
& \times e^{- N_D \left(1 + \frac{\sigma_w^2}{|h_w|^2 P_D}\right) \ln \left( \frac{|h_w|^2 P_D}{\sigma_w^2} + 1 \right) } \notag \\
& \times \left[ \frac{|h_w|^2 P_D - \sigma_w^2 \ln \left( \frac{|h_w|^2 P_D}{\sigma_w^2} + 1 \right) }{|h_w|^2 P_D^2}  \right] .
\end{flalign}

The next step is to apply the limit as $P_D \rightarrow 0$. Thus
\begin{equation}\label{eq25}
\begin{aligned}
&\underset{P_D \rightarrow 0}{\mathrm{\lim}} \: \frac{\partial \zeta_{w, CSI}^*}{\partial P_D} \\
&= \underset{P \rightarrow 0}{\mathrm{\lim}} \: \frac{1}{\Gamma(N_D)} \Bigg[ \frac{\partial \Gamma\left(N_D,  \frac{N_D \sigma_w^2 }{|h_w|^2 P_D}  \ln (\frac{|h_w|^2 P_D}{\sigma_w^2} + 1) \right)}{\partial P_D} \\
& \quad - \frac{\partial \Gamma\left(N_D, N_D \left(1 + \frac{\sigma_w^2}{|h_w|^2 P_D}\right)  \ln (\frac{|h_w|^2 P_D}{\sigma_w^2} + 1) \right)}{\partial P_D} \Bigg] ,
\end{aligned}
\end{equation}
where, using the law of products for limits, we calculate the limit at each factor of the above derivatives separately as follows. \\

\noindent
For the first factor in (\ref{eq23}),

\begin{align}\label{eq26}
\underset{P_D \rightarrow 0}{\mathrm{\lim}} \: &\left(\frac{\sigma_w^2 }{|h_w|^2 P_D}  \ln \left(\frac{|h_w|^2 P_D}{\sigma_w^2} + 1\right) \right)^{N_D-1} \notag \\
&= \left( \underset{P_D \rightarrow 0}{\mathrm{\lim}} \: \frac{\sigma_w^2 }{|h_w|^2 P_D}  \ln \left(\frac{|h_w|^2 P_D}{\sigma_w^2} + 1\right)  \right)^{N_D-1} \notag \\
&= 1^{N_D-1} = 1
\end{align}

where we have used L'Hopital rule to find the internal limit. For the second factor in (\ref{eq23}),

\begin{align}\label{eq27}
\underset{P_D \rightarrow 0}{\mathrm{\lim}} \: &e^{- \frac{N_D \sigma_w^2 }{|h_w|^2 P_D}  \ln \left(\frac{|h_w|^2 P_D}{\sigma_w^2} + 1\right) } \notag \\
&= \underset{P_D \rightarrow 0}{\mathrm{\lim}} \: \left(\frac{|h_w|^2 P_D}{\sigma_w^2} + 1\right)^{-\frac{N_D \sigma_w^2}{|h_w|^2 P_D}} \notag \\
&= \left[ \underset{P_D \rightarrow 0}{\mathrm{\lim}} \: \left(\frac{|h_w|^2 P_D}{\sigma_w^2} + 1\right)^{-\frac{\sigma_w^2}{|h_w|^2 P_D}}  \right]^{N_D} \notag \\
&= \left[ e^{-1} \right]^{N_D} = e^{-N_D}
\end{align}

where we have used the Euler's identity \cite{gradshteyn2014table}, given by
\begin{equation}
e^x = \underset{n \rightarrow \infty}{\mathrm{\lim}} \: \left(1 + \frac{x}{n} \right)^n ,
\end{equation}
to calculate the internal limit. For the third factor in (\ref{eq23}), repeated application of L'Hopital rule yields
\begin{equation}\label{eq28}
\begin{aligned}
&\underset{P_D \rightarrow 0}{\mathrm{\lim}} \: \left[\frac{\sigma_w^2}{P_D (|h_w|^2 P_D + \sigma_w^2)} - \frac{\sigma_w^2 \ln \left( \frac{|h_w|^2 P_D}{\sigma_w^2} + 1 \right)}{|h_w|^2 P_D^2}  \right] \\
&\qquad \qquad \qquad = - \frac{|h_w|^2}{2 \sigma_w^2}.
\end{aligned}
\end{equation}
Hence, overall for the first term on RHS of (\ref{eq25}), we have

\begin{align}\label{eq29}
&\underset{P_D \rightarrow 0}{\mathrm{\lim}} \: \frac{1}{\Gamma(N_D)}\frac{\partial \Gamma\left(N_D,  \frac{N_D \sigma_w^2 }{|h_w|^2 P_D}  \ln (\frac{|h_w|^2 P_D}{\sigma_w^2} + 1) \right)}{\partial P_D} \notag \\
& \qquad \qquad \qquad = - \frac{N_D^{N_D} e^{-N_D} |h_w|^2}{2 \sigma_w^2 \Gamma(N_D)}.
\end{align}

Similarly, for the first factor in (\ref{eq24}),

\begin{align}\label{eq30}
&\underset{P_D \rightarrow 0}{\mathrm{\lim}} \: \left( \left(1 + \frac{\sigma_w^2 }{|h_w|^2 P_D} \right)  \ln \left(\frac{|h_w|^2 P_D}{\sigma_w^2} + 1\right) \right)^{N_D-1} \notag \\
&= \left( \underset{P_D \rightarrow 0}{\mathrm{\lim}} \: \left( 1 + \frac{\sigma_w^2 }{|h_w|^2 P_D} \right) \ln \left(\frac{|h_w|^2 P_D}{\sigma_w^2} + 1\right)  \right)^{N_D-1} \notag \\
&= 1^{N_D-1} = 1
\end{align}

where we have again used L'Hopital rule to find the internal limit. For the second factor in (\ref{eq24}),

\begin{align}\label{eq31}
&\underset{P_D \rightarrow 0}{\mathrm{\lim}} \: e^{- \left( 1 + \frac{N_D \sigma_w^2 }{|h_w|^2 P_D} \right)  \ln \left(\frac{|h_w|^2 P_D}{\sigma_w^2} + 1\right) } \notag \\
&= \underset{P_D \rightarrow 0}{\mathrm{\lim}} \: \left(\frac{|h_w|^2 P_D}{\sigma_w^2} + 1\right)^{- N_D \left( 1 + \frac{\sigma_w^2}{|h_w|^2 P_D} \right)} \notag \\
&= \left[ \underset{P_D \rightarrow 0}{\mathrm{\lim}} \: \left(\frac{|h_w|^2 P_D}{\sigma_w^2} + 1\right)^{- \left( 1 + \frac{\sigma_w^2}{|h_w|^2 P_D} \right) }  \right]^{N_D} \notag \\
&= \left[ e^{-1} \right]^{N_D} = e^{-N_D}
\end{align}

where we have again used the Euler's identity to calculate the internal limit. For the third factor in (\ref{eq24}),
\begin{equation}\label{eq32}
\begin{aligned}
\underset{P_D \rightarrow 0}{\mathrm{\lim}} \: \left[ \frac{|h_w|^2 P_D - \sigma_w^2 \ln \left( \frac{|h_w|^2 P_D}{\sigma_w^2} + 1 \right) }{|h_w|^2 P_D^2}   \right] = \frac{|h_w|^2}{2 \sigma_w^2}.
\end{aligned}
\end{equation}
Hence, overall for the second term on RHS of (\ref{eq25}), we have
\begin{equation}\label{eq33}
\begin{aligned}
&\underset{P_D \rightarrow 0}{\mathrm{\lim}} \: \frac{1}{\Gamma(N_D)} \frac{\partial \Gamma\left(N_D, N_D \left(1 + \frac{\sigma_w^2}{|h_w|^2 P_D}\right)  \ln (\frac{|h_w|^2 P_D}{\sigma_w^2} + 1) \right)}{\partial P_D} \\
& \qquad \qquad \qquad = \frac{N_D^{N_D} e^{-N_D} |h_w|^2}{2 \sigma_w^2 \Gamma(N_D)}.
\end{aligned}
\end{equation}
Combining the results in (\ref{eq29}) and (\ref{eq33}), we have
\begin{equation}\label{eq34}
\underset{P_D \rightarrow 0}{\mathrm{\lim}} \: \frac{\partial \zeta_{w, CSI}^*}{\partial P_D} = - \frac{N_D^{N_D} e^{-N_D} |h_w|^2}{\sigma_w^2 \Gamma(N_D)},
\end{equation}
which is the slope of the first order approximation, hence completing the proof.

\section*{Appendix B \\ Proof of Lemma 3}
The problem at Willie is of finding $\zeta_{w, CDI}^*$, given by
\begin{equation}\label{eq36}
\zeta_{w, CDI}^* = \underset{\lambda}{\mathrm{min}} \quad \mathbb{E}_{|h_w|^2} \left[ \zeta_{w, CDI} \right] .
\end{equation}
Using the relationship of incomplete and complete Gamma functions given by
\begin{equation}\label{eq37}
\Gamma(a) = \Gamma(a,b) + \gamma(a,b),
\end{equation}
we can rewrite $\zeta_{w, CDI}$ of (\ref{eq12}) as
\begin{equation}\label{eq38}
\zeta_{w, CDI} = 1 + \frac{\Gamma\left(N_D,  \frac{N_D \lambda}{\sigma_w^2} \right)}{\Gamma(N_D)} - \frac{\Gamma\left(N_D,  \frac{N_D \lambda}{|h_w|^2 P_D + \sigma_w^2} \right)}{\Gamma(N_D)}.
\end{equation}
Here, we consider a linear approximation of $\zeta_{w, CDI}$ using Taylor series expansion. where the first two terms of the expansion around $P_D=0$ are considered, and these two terms are given by $\left[ f(0) + Pf'(0) \right]$, where $f(P_D)$ is given by (\ref{eq38}) above. We first note that here, $f(0)=1$. To calculate the derivative of $f(P_D)$, we use the derivative property of upper incomplete Gamma function and the required derivative is calculated as

\begin{align}\label{eq40}
 \frac{\partial f(P_D)}{\partial P_D} &= - \Bigg[ -\frac{1}{\Gamma(N_D)} \left( \frac{N_D \lambda}{|h_w|^2 P_D + \sigma_w^2} \right)^{N_D-1} \notag \\
 & \quad \times e^{- \frac{N_D \lambda}{|h_w|^2 P_D + \sigma_w^2}} \left( - \frac{N_D \lambda |h_w|^2}{(|h_w|^2 P_D + \sigma_w^2)^2}  \right) \Bigg] \notag \\
&= - \frac{1}{\Gamma(N_D)} \left( \frac{N_D \lambda |h_w|^2}{(|h_w|^2 P_D + \sigma_w^2)^2}  \right) \notag \\
& \quad \times \left( \frac{N_D \lambda}{|h_w|^2 P_D + \sigma_w^2} \right)^{N_D-1} e^{- \frac{N_D \lambda}{|h_w|^2 P_D + \sigma_w^2}}
\end{align}

which for $P_D=0$ becomes

\begin{align}\label{eq41}
\frac{\partial f(P_D)}{\partial P_D} \Bigg|_{P_D=0} = &- \frac{1}{\Gamma(N_D)} \left( \frac{N_D \lambda |h_w|^2}{(\sigma_w^2)^2}  \right) \notag \\
& \times \left( \frac{N_D \lambda}{\sigma_w^2} \right)^{N_D-1} e^{- \frac{N_D \lambda}{\sigma_w^2}} .
\end{align}

Hence, we have the linear approximation for $\zeta_{w, CDI}$ as
\begin{equation}\label{eq42}
\zeta_{w, CDI} \approx 1 - \frac{P_D}{\Gamma(N_D)} \left( \frac{N_D \lambda |h_w|^2}{(\sigma_w^2)^2}  \right) \left( \frac{N_D \lambda}{\sigma_w^2} \right)^{N_D-1} e^{- \frac{N_D \lambda}{\sigma_w^2}} .
\end{equation}
To find the best threshold for Willie under this approximation, we consider
\begin{equation}\label{eq43}
\lambda_{CDI}^* = \underset{\lambda}{\mathrm{ \arg \min}} \: \mathbb{E}_{|h_w|^2} \left[ \zeta_{w, CDI} \right] ,
\end{equation}
where due to $\mathbb{E}\left[ |h_w|^2 \right] = 1$, we have
\begin{equation}\label{eq44}
\mathbb{E}_{|h_w|^2} \left[ \zeta_{w, CDI} \right] \approx 1 - \left( \frac{N_D \lambda P_D}{(\sigma_w^2)^2 \Gamma(N_D)}  \right) \left( \frac{N_D \lambda}{\sigma_w^2} \right)^{N_D-1} e^{- \frac{N_D \lambda}{\sigma_w^2}}
\end{equation}
Differentiating this quantity w.r.t $\lambda$ gives

\begin{flalign}\label{eq45}
&\frac{\partial \mathbb{E}_{|h_w|^2} \left[ \zeta_{w, CDI} \right] }{\partial \lambda} \notag \\
&= - \frac{N_D^{N_D} P_D}{\Gamma(N_D) \left(\sigma_w^2  \right)^{N_D+1}} \Bigg[N_D \lambda^{N_D-1}e^{- \frac{N_D \lambda}{\sigma_w^2}} \notag \\
& \qquad \qquad \qquad \qquad \qquad - \frac{N_D \lambda^{N_D}}{\sigma_w^2} e^{- \frac{N_D \lambda}{\sigma_w^2}}  \Bigg].
\end{flalign}

Setting the above derivative equal to zero and some further simplifications give
\begin{equation}\label{eq46}
\lambda_{CDI}^* = \sigma_w^2 ,
\end{equation}
Using this value of $\lambda_{CDI}^*$ in the linear approximation of $\zeta_{w, CDI}$ completes the proof.

%\bibliographystyle{IEEETran}
%\bibliography{IEEEabrv,cc_bib_khurram}

\end{document}